\documentclass[journal]{IEEEtran}
\usepackage{array}
\IEEEoverridecommandlockouts
\usepackage{cite}
\usepackage{amsmath,amssymb,amsfonts}
\usepackage{algorithm}
\usepackage{algpseudocode}
\usepackage{graphicx}
\usepackage{amsmath}
\usepackage{amsthm}
\usepackage{textcomp}
\usepackage{xcolor}
\usepackage{amsmath}
\usepackage{float} 
\usepackage{comment}
\usepackage{subfigure} 
\usepackage{color}
\usepackage{bm}
\usepackage{esvect}

\usepackage[top=0.73in, bottom=0.68in, left=0.58in, right=0.58in]{geometry}

\def\BibTeX{{\mathrm B\kern-.05em{\sci\kern-.025em b}\kern-.08em
    T\kern-.1667em\lower.7ex\hbox{E}\kern-.125emX}}
\begin{document}
\newtheorem{theorem}{\emph{\underline{Theorem}}}
\newtheorem{acknowledgement}[theorem]{Acknowledgement}
\renewcommand{\algorithmicensure}{ \textbf{repeat:}}
\newtheorem{axiom}[theorem]{Axiom}
\newtheorem{case}[theorem]{Case}
\newtheorem{claim}[theorem]{Claim}
\newtheorem{conclusion}[theorem]{Conclusion}
\newtheorem{condition}[theorem]{Condition}
\newtheorem{conjecture}[theorem]{Conjecture}
\newtheorem{criterion}[theorem]{Criterion}
\newtheorem{definition}{Definition}
\newtheorem{exercise}[theorem]{Exercise}
\newtheorem{lemma}{\emph{\underline{Lemma}}}
\newtheorem{corollary}{\emph{\underline{Corollary}}}
\newtheorem{notation}[theorem]{Notation}
\newtheorem{problem}[theorem]{Problem}
\newtheorem{proposition}{\emph{\underline{Proposition}}}
\newtheorem{solution}[theorem]{Solution}
\newtheorem{summary}[theorem]{Summary}
\newtheorem{assumption}{Assumption}
\newtheorem{example}{\bf Example}
\newtheorem{remark}{\bf \emph{\underline{Remark}}}

\def\qed{$\Box$}
\def\QED{\mbox{\phantom{m}}\nolinebreak\hfill$\,\Box$}
\def\proof{\noindent{\emph{Proof:} }}
\def\poof{\noindent{\emph{Sketch of Proof:} }}
\def
\endproof{\hspace*{\fill}~\qed
\par
\endtrivlist\unskip}
\def\endproof{\hspace*{\fill}~\qed\par\endtrivlist\vskip3pt}

\def\E{\mathsf{E}}
\def\eps{\varepsilon}
\def\phi{\varphi}
\def\Lsp{{\boldsymbol L}}
\def\Bsp{{\boldsymbol B}}
\def\lsp{{\boldsymbol\ell}}
\def\Ltsp{{\Lsp^2}}
\def\Lpsp{{\Lsp^p}}
\def\Linsp{{\Lsp^{\infty}}}
\def\LtR{{\Lsp^2(\Rst)}}
\def\ltZ{{\lsp^2(\Zst)}}
\def\ltsp{{\lsp^2}}
\def\ltZt{{\lsp^2(\Zst^{2})}}
\def\ninN{{n{\in}\Nst}}
\def\oh{{\frac{1}{2}}}
\def\grass{{\cal G}}
\def\ord{{\cal O}}
\def\dist{{d_G}}
\def\conj#1{{\overline#1}}
\def\ntoinf{{n \rightarrow \infty }}
\def\toinf{{\rightarrow \infty }}
\def\tozero{{\rightarrow 0 }}
\def\trace{{\operatorname{trace}}}
\def\ord{{\cal O}}
\def\UU{{\cal U}}
\def\rank{{\operatorname{rank}}}
\def\acos{{\operatorname{acos}}}

\def\SINR{\mathsf{SINR}}
\def\SNR{\mathsf{SNR}}
\def\SIR{\mathsf{SIR}}
\def\tSIR{\widetilde{\mathsf{SIR}}}
\def\Ei{\mathsf{Ei}}
\def\l{\left}
\def\r{\right}
\def\({\left(}
\def\){\right)}
\def\lb{\left\{}
\def\rb{\right\}}

\setcounter{page}{1}

\newcommand{\eref}[1]{(\ref{#1})}
\newcommand{\fig}[1]{Fig.\ \ref{#1}}

\def\bydef{:=}
\def\ba{{\mathbf{a}}}
\def\bb{{\mathbf{b}}}
\def\bc{{\mathbf{c}}}
\def\bd{{\mathbf{d}}}
\def\bee{{\mathbf{e}}}
\def\bff{{\mathbf{f}}}
\def\bg{{\mathbf{g}}}
\def\bh{{\mathbf{h}}}
\def\bi{{\mathbf{i}}}
\def\bj{{\mathbf{j}}}
\def\bk{{\mathbf{k}}}
\def\bl{{\mathbf{l}}}
\def\bm{{\mathbf{m}}}
\def\bn{{\mathbf{n}}}
\def\bo{{\mathbf{o}}}
\def\bp{{\mathbf{p}}}
\def\bq{{\mathbf{q}}}
\def\br{{\mathbf{r}}}
\def\bs{{\mathbf{s}}}
\def\bt{{\mathbf{t}}}
\def\bu{{\mathbf{u}}}
\def\bv{{\mathbf{v}}}
\def\bw{{\mathbf{w}}}
\def\bx{{\mathbf{x}}}
\def\by{{\mathbf{y}}}
\def\bz{{\mathbf{z}}}
\def\b0{{\mathbf{0}}}

\def\bA{{\mathbf{A}}}
\def\bB{{\mathbf{B}}}
\def\bC{{\mathbf{C}}}
\def\bD{{\mathbf{D}}}
\def\bE{{\mathbf{E}}}
\def\bF{{\mathbf{F}}}
\def\bG{{\mathbf{G}}}
\def\bH{{\mathbf{H}}}
\def\bI{{\mathbf{I}}}
\def\bJ{{\mathbf{J}}}
\def\bK{{\mathbf{K}}}
\def\bL{{\mathbf{L}}}
\def\bM{{\mathbf{M}}}
\def\bN{{\mathbf{N}}}
\def\bO{{\mathbf{O}}}
\def\bP{{\mathbf{P}}}
\def\bQ{{\mathbf{Q}}}
\def\bR{{\mathbf{R}}}
\def\bS{{\mathbf{S}}}
\def\bT{{\mathbf{T}}}
\def\bU{{\mathbf{U}}}
\def\bV{{\mathbf{V}}}
\def\bW{{\mathbf{W}}}
\def\bX{{\mathbf{X}}}
\def\bY{{\mathbf{Y}}}
\def\bZ{{\mathbf{Z}}}

\def\mA{{\mathbb{A}}}
\def\mB{{\mathbb{B}}}
\def\mC{{\mathbb{C}}}
\def\mD{{\mathbb{D}}}
\def\mE{{\mathbb{E}}}
\def\mF{{\mathbb{F}}}
\def\mG{{\mathbb{G}}}
\def\mH{{\mathbb{H}}}
\def\mI{{\mathbb{I}}}
\def\mJ{{\mathbb{J}}}
\def\mK{{\mathbb{K}}}
\def\mL{{\mathbb{L}}}
\def\mM{{\mathbb{M}}}
\def\mN{{\mathbb{N}}}
\def\mO{{\mathbb{O}}}
\def\mP{{\mathbb{P}}}
\def\mQ{{\mathbb{Q}}}
\def\mR{{\mathbb{R}}}
\def\mS{{\mathbb{S}}}
\def\mT{{\mathbb{T}}}
\def\mU{{\mathbb{U}}}
\def\mV{{\mathbb{V}}}
\def\mW{{\mathbb{W}}}
\def\mX{{\mathbb{X}}}
\def\mY{{\mathbb{Y}}}
\def\mZ{{\mathbb{Z}}}

\def\cA{\mathcal{A}}
\def\cB{\mathcal{B}}
\def\cC{\mathcal{C}}
\def\cD{\mathcal{D}}
\def\cE{\mathcal{E}}
\def\cF{\mathcal{F}}
\def\cG{\mathcal{G}}
\def\cH{\mathcal{H}}
\def\cI{\mathcal{I}}
\def\cJ{\mathcal{J}}
\def\cK{\mathcal{K}}
\def\cL{\mathcal{L}}
\def\cM{\mathcal{M}}
\def\cN{\mathcal{N}}
\def\cO{\mathcal{O}}
\def\cP{\mathcal{P}}
\def\cQ{\mathcal{Q}}
\def\cR{\mathcal{R}}
\def\cS{\mathcal{S}}
\def\cT{\mathcal{T}}
\def\cU{\mathcal{U}}
\def\cV{\mathcal{V}}
\def\cW{\mathcal{W}}
\def\cX{\mathcal{X}}
\def\cY{\mathcal{Y}}
\def\cZ{\mathcal{Z}}
\def\cd{\mathcal{d}}
\def\Mt{M_{t}}
\def\Mr{M_{r}}
\def\O{\Omega_{M_{t}}}
\newcommand{\figref}[1]{{Fig.}~\ref{#1}}
\newcommand{\tabref}[1]{{Table}~\ref{#1}}

\newcommand{\var}{\mathsf{var}}
\newcommand{\fb}{\tx{fb}}
\newcommand{\nf}{\tx{nf}}
\newcommand{\BC}{\tx{(bc)}}
\newcommand{\MAC}{\tx{(mac)}}
\newcommand{\Pout}{p_{\mathsf{out}}}
\newcommand{\nnn}{\nn\\}
\newcommand{\FB}{\tx{FB}}
\newcommand{\TX}{\tx{TX}}
\newcommand{\RX}{\tx{RX}}
\renewcommand{\mod}{\tx{mod}}
\newcommand{\m}[1]{\mathbf{#1}}
\newcommand{\td}[1]{\tilde{#1}}
\newcommand{\sbf}[1]{\scriptsize{\textbf{#1}}}
\newcommand{\stxt}[1]{\scriptsize{\textrm{#1}}}
\newcommand{\suml}[2]{\sum\limits_{#1}^{#2}}
\newcommand{\sumlk}{\sum\limits_{k=0}^{K-1}}
\newcommand{\eqhsp}{\hspace{10 pt}}
\newcommand{\tx}[1]{\texttt{#1}}
\newcommand{\Hz}{\ \tx{Hz}}
\newcommand{\sinc}{\tx{sinc}}
\newcommand{\tr}{\mathrm{tr}}
\newcommand{\diag}{\mathrm{diag}}
\newcommand{\MAI}{\tx{MAI}}
\newcommand{\ISI}{\tx{ISI}}
\newcommand{\IBI}{\tx{IBI}}
\newcommand{\CN}{\tx{CN}}
\newcommand{\CP}{\tx{CP}}
\newcommand{\ZP}{\tx{ZP}}
\newcommand{\ZF}{\tx{ZF}}
\newcommand{\SP}{\tx{SP}}
\newcommand{\MMSE}{\tx{MMSE}}
\newcommand{\MINF}{\tx{MINF}}
\newcommand{\RC}{\tx{MP}}
\newcommand{\MBER}{\tx{MBER}}
\newcommand{\MSNR}{\tx{MSNR}}
\newcommand{\MCAP}{\tx{MCAP}}
\newcommand{\vol}{\tx{vol}}
\newcommand{\ah}{\hat{g}}
\newcommand{\tg}{\tilde{g}}
\newcommand{\teta}{\tilde{\eta}}
\newcommand{\heta}{\hat{\eta}}
\newcommand{\uh}{\m{\hat{s}}}
\newcommand{\eh}{\m{\hat{\eta}}}
\newcommand{\hv}{\m{h}}
\newcommand{\hh}{\m{\hat{h}}}
\newcommand{\Po}{P_{\mathrm{out}}}
\newcommand{\Poh}{\hat{P}_{\mathrm{out}}}
\newcommand{\Ph}{\hat{\gamma}}
\newcommand{\mat}[1]{\begin{matrix}#1\end{matrix}}
\newcommand{\ud}{^{\dagger}}
\newcommand{\C}{\mathcal{C}}
\newcommand{\nn}{\nonumber}
\newcommand{\nInf}{U\rightarrow \infty}

\title{
\begin{flushleft}\vspace{-2cm}
\author{Zhenyu~Kang,~\IEEEmembership{Graduate~Student~Member,~IEEE}, Changsheng~You,~\IEEEmembership{Member,~IEEE},\\
	 and Rui Zhang,~\IEEEmembership{Fellow,~IEEE}  \thanks{\noindent Z. Kang is with the Department of Electrical and Computer Engineering, National University of Singapore, Singapore 117583 (email: zhenyu\_kang@u.nus.edu). 
C. You is with the Department of Electronic and Electrical Engineering, Southern University of Science and Technology (SUSTech), Shenzhen 518055, China (e-mail: youcs@sustech.edu.cn).
R. Zhang is with School of Science and Engineering, Shenzhen Research Institute of Big Data, The Chinese University of Hong Kong, Shenzhen, Guangdong 518172, China (e-mail: rzhang@cuhk.edu.cn). He is also with the Department of Electrical and Computer Engineering, National University of Singapore, Singapore 117583 (e-mail: elezhang@nus.edu.sg). (Corresponding author: Rui Zhang and Changsheng You).

This work was supported in part by the 2022 Stable Research Program of Higher Education of China under Grant 20220817144726001, the National Natural Science Foundation of China under Grant 62201242, the Ministry of Education, Singapore under Award T2EP50120-0024, the Advanced Research and Technology Innovation Centre (ARTIC) of National University of Singapore under Research Grant R-261-518-005-720, and The Guangdong Provincial Key Laboratory of Big Data Computing.
}
  \vspace{-1cm}}
\end{flushleft}
\huge 
Double-Active-IRS Aided Wireless Communication: Deployment Optimization and Capacity Scaling} 
\maketitle
\begin{abstract}
In this letter, {\color{black}we consider a \textit{double}-active-intelligent reflecting surface (IRS) aided wireless communication system, where two active IRSs are properly deployed to assist the communication from a base station (BS) to multiple users located in a given zone via the double-reflection links.}
Under the assumption of fixed \textit{per-element} amplification power for each active-IRS element, we formulate a rate maximization problem subject to practical constraints on the reflection design, elements allocation, and placement of active IRSs.
{\color{black}To solve this non-convex problem, we first obtain the optimal active-IRS reflections and BS beamforming, based on which we then jointly optimize the active-IRS elements allocation and placement by using the alternating optimization (AO) method.}
Moreover, we show that given the fixed per-element amplification power, the received signal-to-noise ratio (SNR) at the user increases asymptotically with the \textit{square} of the number of reflecting elements; while given the fixed number of reflecting elements, the SNR does not increase with the per-element amplification power when it is asymptotically large.
Last, numerical results are presented to validate the effectiveness of the proposed AO-based algorithm and compare the rate performance of the considered double-active-IRS aided wireless system with various benchmark systems.

\end{abstract}
\vspace{-15pt}
\section{Introduction}
\renewcommand{\baselinestretch}{1}
Intelligent reflecting surface (IRS) has emerged as a promising technology to smartly reconfigure the wireless radio propagation environment \cite{9326394,9591503}.
Specifically, IRS is a low-cost meta-surface consisting of massive reflecting elements that can reflect incident signals with flexibly tuned phase shifts and/or amplitudes to enhance desired signal power or suppress undesired interference \cite{9424177}.
However, the conventional IRS with fully passive reflecting elements suffers severe product-distance path loss in practice, which significantly limits the power of IRS-reflected signals.

To address this issue, a new type of IRS, called \textit{active} IRS, has been recently proposed, which enables simultaneous signal reflection and amplification by using reflection-type amplifiers, hence more effectively compensating the severe path loss of passive IRS \cite{9377648,9944635}.
Specifically, it was shown in \cite{9424177,9734027} that given the same IRS location, active IRS achieves a higher rate than passive IRS thanks to the appealing amplification gain.
Besides, the authors in \cite{9530750} showed that active IRS should be properly deployed between the transmitter and receiver to balance the trade-off between the signal and noise amplification, which is in sharp contrast to the case of passive IRS that should be deployed near the transmitter or receiver to minimize the cascaded channel path loss.
However, existing works on active IRS have ignored the cooperation between them, which, however, has the potential to achieve higher channel capacity than the single-active-IRS system, as shown in the case of passive IRS \cite{9362274,9060923,9241706}.
Generally speaking, the design for double-active-IRS aided systems is more complicated than the single-active-IRS counterpart, as elaborated next.
First, the active-IRS placement needs to be carefully devised to balance the path loss of different IRS-related channels.
Second, it is necessary to properly assign the total reflecting elements to the two IRSs to balance the multiplicative beamforming gain and the amplification noise power given the fixed total number of elements.
To the authors' best knowledge, the design of efficient double-active-IRS deployment and its rate performance comparison with the single-active-IRS counterpart have not been studied yet.

To answer the above questions, {\color{black}we study in this letter a double-active-IRS aided wireless communication system as illustrated in Fig. \ref{sysmod}, where two active IRSs are deployed to assist the communication between a multi-antenna base station (BS) and multiple single-antenna users.}
Note that unlike existing works that mostly considered the total active-IRS amplification power constraint (e.g., \cite{9377648,9530750,9734027,9944635}), we consider in this work the new \textit{per-element} amplification power constraint for each active-IRS reflecting element\cite{6019001}.\footnote{{\color{black}Note that although the active IRS with per-element amplification power control requires higher hardware cost than that with the total amplification power, it leads to lower computational complexity, more flexible control, and better fault tolerance.}} 
{\color{black}Thereby, we formulate an optimization problem to maximize the achievable rate of the double-active-IRS system subject to practical constraints on the reflection and beamforming design, elements allocation, and placement of active IRSs.}
To solve this non-convex problem, {\color{black}we first obtain the optimal active-IRS reflections and BS beamforming}, based on which we then jointly optimize the active-IRS elements allocation and placement by using the alternating optimization (AO) method.
Moreover, we analytically characterize the system capacity scaling orders with respect to (w.r.t.) the number of reflecting elements and the per-element amplification power.
Last, numerical results are provided to evaluate the proposed algorithm and compare the rate performance of the double-active-IRS aided wireless system with various benchmark systems.

\begin{figure}[t]
\centerline{\includegraphics[width=2.5in]{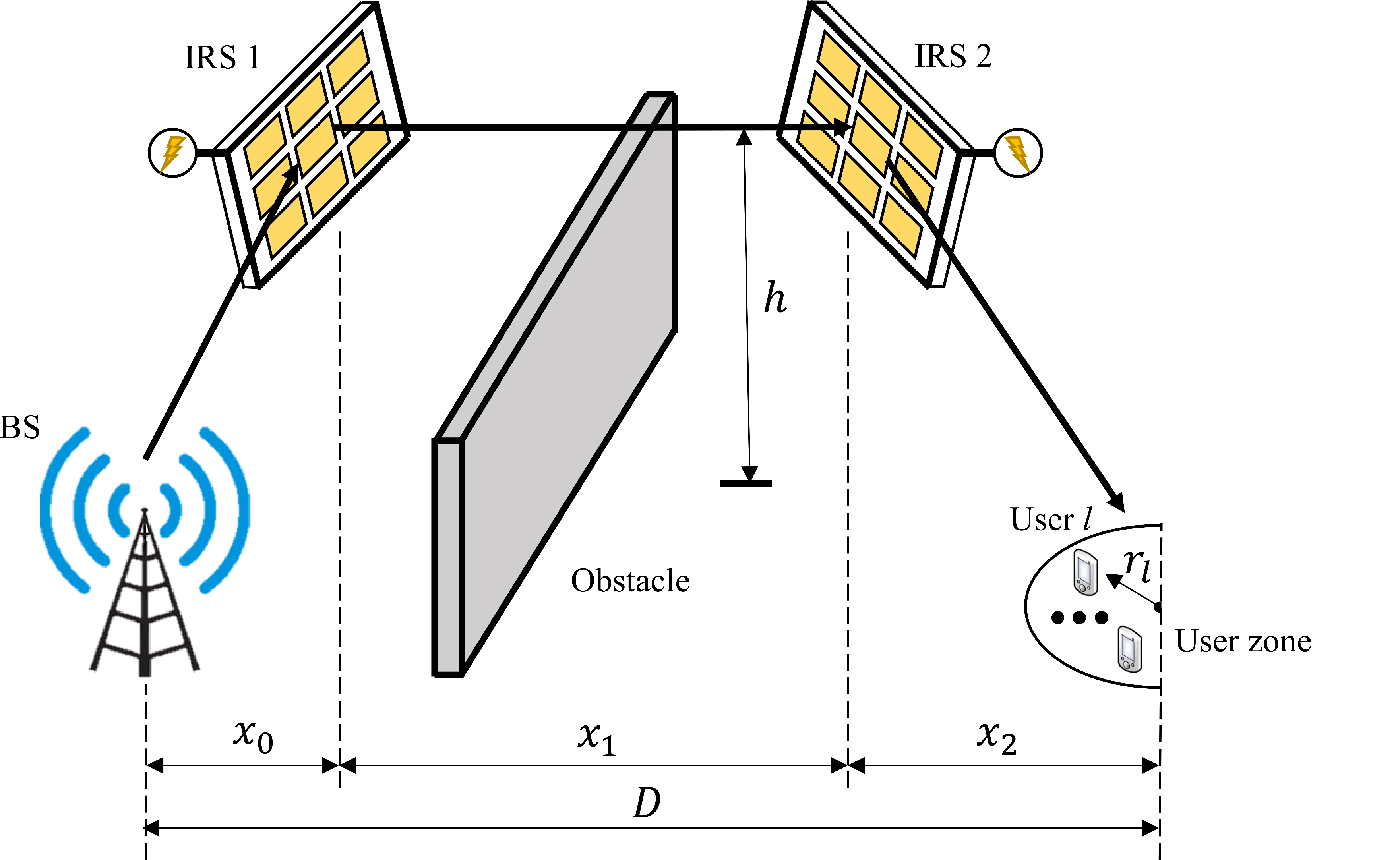}}
\caption{{\color{black}The double-active-IRS aided wireless communication system.}}\label{sysmod}
\vspace{-0.7cm}
\end{figure}

\vspace{-0.2cm}
\section{System Model and Problem Formulation}
\subsection{System Model}
\vspace{-2pt}
{\color{black}Consider a double-active-IRS aided wireless communication system as shown in Fig. \ref{sysmod}, where two active IRSs are properly deployed to assist the communication from a multi-antenna BS to $\ell$ single-antenna users in a given zone.} We consider a challenging scenario where the direct BS-user and single-reflection links are blocked; thus the BS can communicate with the user via a double-reflection link only, i.e., BS$\to$IRS 1$\to$IRS 2$\to$user.
{\color{black}We consider the time-division multiple access (TDMA) schemes, where the BS serves different users in different time\footnote{{\color{black}As the first study on double-active-IRS deployment, this work aims to provide useful insights into its deployment and system capacity scaling order. The active-IRS beamforming design for other multi-user cases, e.g., non-orthogonal multiple access (NOMA), will be studied in our future work.}}.}
{\color{black}Without loss of generality, we consider a three-dimensional (3D) Cartesian coordinate system, where the BS, the user-zone center, and two IRSs are located at $\boldsymbol{u}_{\rm B}=(0,0,0)$, $\boldsymbol{u}_{\rm u}=(D,0,0)$, $\boldsymbol{u}_{1}=(x_0,0,H)$ and $\boldsymbol{u}_{2}=(x_0+x_1,0,H)$, respectively, where the horizontal locations of the two IRSs (or equivalently $x_0$, $x_1$) need to be designed with $r_{\ell}$ denoting the distance between the center of user zone and $\ell$-th user and $\theta_{\ell}$ denoting its azimuth angle, $\ell\in\{1,\cdots,L\}$.} Note that in this letter, we do not consider dynamic IRS deployment, since in practice, IRSs are usually at fixed locations once deployed.
\subsubsection{Active-IRS Model}
Let $M$ denote the deployment budget on the total number of reflecting elements with $M_1$ and $M_2$ representing the numbers of reflecting elements at IRSs 1 and 2, respectively.
{\color{black}For each active IRS $k$ in the $\ell$-th time slot with $k\in\{1,2\}$, let $\boldsymbol{\Psi}_{k,\ell}\triangleq\boldsymbol{A}_{k,\ell}\boldsymbol{\Phi}_{k,\ell}$ denote its reflection matrix where $\boldsymbol{A}_{k,\ell}\triangleq\diag(a_{\mathrm I_k,\ell,1},\cdots,a_{\mathrm I_k,\ell,M_k})$ and $\boldsymbol{\Phi}_{k,\ell}\triangleq\diag(e^{\jmath\phi_{\mathrm I_k,\ell,1}},\cdots,e^{\jmath\phi_{\mathrm I_k,\ell,M_k}})$ denote the amplification matrix and phase-shift matrix. Herein, $a_{\mathrm I_k,\ell,m_k}$ and $e^{\jmath\phi_{\mathrm I_k,\ell,m_k}}$ represent the amplification factor and phase shift of each element.}
We consider the fixed per-element amplification power for the two active IRSs, where a separate power supply of $P_{\rm e}$ is connected to each reflecting element for signal amplification. 
\subsubsection{Channel Model}
For ease of analysis, {\color{black}we consider the line-of-sight (LoS) channel model for all available links, for which the channel state information (CSI) on involved links are assumed to be known\footnote{{\color{black}The results obtained in this letter can be extended to more general cases with non-LoS paths, e.g., the Rician fading channel. In practice, the effect of channel estimation error can be mitigated by designing robust IRS beamforming and deployment based on the distribution of channel estimation error}.}. Note that the design of active IRS elements allocation and placement mainly depends on the LoS channel path loss, which can be obtained based on the locations of the BS, active IRSs and user.}
Let $\theta_{i,j}^{\mathrm{a}}(\vartheta_{i,j}^{\mathrm{a}})\in[0,\pi]$ denote the azimuth (elevation) angle-of-arrival (AoA) at node $j$ from node $i$, $\theta_{i,j}^{\mathrm{a}}(\vartheta_{i,j}^{\mathrm{d}})\in[0,\pi]$ denote the azimuth (elevation) angle-of-departure (AoD) from node $i$ to node $j$, and $\boldsymbol{a}_{\mathrm{r}}$ denote the receive steering vector, given by $\boldsymbol{a}_{\mathrm{r}}\left(\theta_{i,j}^{\mathrm{a}}, \vartheta_{i,j}^{\mathrm{a}}, N_j\right)\triangleq\boldsymbol{w}\!\left(\frac{2 d_{\mathrm{I}}}{\lambda} \!\cos \theta_{i,j}^{\mathrm{a}} \sin \vartheta_{i,j}^{\mathrm{a}}, N_{j, \mathrm{h}}\right)\!\otimes\!\boldsymbol{w}\!\left(\frac{2 d_{\mathrm{I}}}{\lambda} \!\cos \vartheta_{i,j}^{\mathrm{a}}, N_{j, \mathrm{v}}\right)$ with $\otimes$ denoting the Kronecker product, $d_{\mathrm{I}}$ denoting the IRS reflecting element spacing and $\lambda$ denoting the signal wavelength, $N_{j, \mathrm{h}}$ ($N_{j, \mathrm{v}}$) denoting the number of horizontal (vertical) elements of node $j$, and $\boldsymbol{w}(\varsigma, N) \!\!\triangleq\!\!\left[1, \cdots, e^{-\jmath \pi(N-1) \varsigma}\right]^T$ denoting the steering vector function. The transmit steering vector, $\boldsymbol{a}_{\mathrm{t}}$, can be modeled similar to $\boldsymbol{a}_{\mathrm{r}}$.
As such, the inter-IRS channel from IRS 1 to IRS 2, denoted by $\boldsymbol{G}\in\mathbb{C}^{M_2\times M_1}$, can be modeled as
\begin{equation}\vspace{-0.1cm}
    \boldsymbol{G}\triangleq g\boldsymbol{a}_{\mathrm{r}}\left(\theta_{\mathrm I_1, \mathrm I_2}^{\mathrm{a}}, \vartheta_{\mathrm I_1, \mathrm I_2}^{\mathrm{a}}, M_2\right) \boldsymbol{a}_{\mathrm{t}}^H\left(\theta_{\mathrm I_1, \mathrm I_2}^{\mathrm{d}}, \vartheta_{\mathrm I_1, \mathrm I_2}^{\mathrm{d}}, M_1\right),\vspace{-0cm}
\end{equation}
where $g\!\!=\!\!\beta^{\frac{1}{2}} e^{-\jmath \frac{2 \pi}{\lambda} d_1} / d_1^{\frac{\alpha}{2}}$ denotes the complex channel gain of the inter-IRS link with $\beta$ denoting the channel power gain at a reference distance of 1 meter (m), $d_1=x_1$ denoting the distance from IRS 1 to IRS 2, and $\alpha$ denoting the path loss exponent.
{\color{black}Let $\boldsymbol{\omega}_{\ell}\in \mathbb{C}^{N \times 1}$ denote the normalized transmit beamforming vector of the BS to $\ell$-th user with $N$ denoting the number of BS antennas and $\|\boldsymbol{\omega}_{\ell}\|^2 \leq 1$.} Then, the channel from the BS to IRS 1, denoted by $\boldsymbol{H}_1\in\mathbb{C}^{M_1 \times N}$, can be modeled as $\boldsymbol{N}_1\triangleq h_1\boldsymbol{a}_{\mathrm{r}}\left(\theta_{\rm BS, \rm I_1}^{\mathrm{a}}, \vartheta_{\rm BS, \rm I_1}^{\mathrm{a}}, M_1\right)\boldsymbol{a}_{\mathrm{t}}^H\left(\theta_{\rm BS, \rm I_1}^{\mathrm{d}}, \vartheta_{\rm BS, \rm I_1}^{\mathrm{d}}, N\right)$,
where its complex channel gain is $h_1=\beta^{\frac{1}{2}} e^{-\jmath \frac{2 \pi}{\lambda} d_0} / d_0^{\frac{\alpha}{2}}$ with $d_0 = \sqrt{x_0^2+h^2}$ denoting the BS-IRS 1 distance. 
{\color{black}The channel from IRS 2 to the $\ell$-th user, denoted by $\boldsymbol{h}_{2,{\ell}}^H\in\mathbb{C}^{1 \times M_2}$, can be modeled in a similar form as $\boldsymbol{h}_1$ with the IRS 2-user $\ell$ distance given by $d_{2,\ell} = \sqrt{(x_2+r_{\ell}\cos{\theta_{\ell}})^2+h^2+(r_{\ell}\sin{\theta_{\ell}})^2}$.}
As such, the received signal at the user is obtained as\vspace{-5pt}
\begin{equation}
    {\color{black}\!\!y_{\ell}\! \!=\!\!\boldsymbol{h}_{2,{\ell}}^H\!\boldsymbol{\Psi}_{\!2,{\ell}}\boldsymbol{G}\boldsymbol{\Psi}_{\!}\boldsymbol{H}_{\!1}\boldsymbol{\omega}_{\ell}s\!\!+\!\!\boldsymbol{h}_{2,{\ell}}^H\boldsymbol{\Psi}_{\!2,{\ell}}\boldsymbol{G}\boldsymbol{\Psi}_{\!1,{\ell}}\boldsymbol{z}_1\!\!+\!\!\boldsymbol{h}_{2,{\ell}}^H\boldsymbol{\Psi}_{\!2,{\ell}}\boldsymbol{z}_2\!\!+\!\!z_0,}
\end{equation}
where $s\in\mathbb{C}$ denotes the transmitted signal with power $P_{\mathrm{B}}$, $\boldsymbol{z}_k\in\mathbb{C}^{M_k\times 1}$ is the amplification noise induced by IRS $k$ that is assumed to follow the independent circularly symmetric complex Gaussian (CSCG) distribution with mean of zero and variance of $\sigma_{\mathrm I}^2$, i.e., $\boldsymbol{z}_{k} \sim \mathcal{C N}\left(\mathbf{0}_{M_k}, \sigma_{\mathrm{I}}^2 \mathbf{I}_{M_k}\right)$ with $\sigma^2_{\rm I}$ denoting the amplification noise power, and $z_0\sim \mathcal{CN}\left(0, \sigma_0^2 \right)$ is the additive white Gaussian noise at the user with power
$\sigma_0^2$. Note that the received signal at the user is superposed by the desired signal over the double-reflection link (i.e., BS$\to$IRS 1$\to$IRS 2$\to$user), the amplification noise induced by IRS 1 over the IRS 1$\to$IRS 2$\to$user link as well as that induced by IRS 2 over the IRS 2$\to$user link.
{\color{black}As such, the corresponding received signal-to-noise ratio (SNR) at the $\ell$-th user, $\gamma_{\ell}$, is given by\vspace{-0.2cm}
\begin{equation}\label{snr_ori}\vspace{-0cm}
    \gamma_{\ell} = \frac{P_{\mathrm{B}}\left|\boldsymbol{h}_{2,{\ell}}^H\boldsymbol{\Psi}_{2,{\ell}}\boldsymbol{G}\boldsymbol{\Psi}_{1,{\ell}}\boldsymbol{H}_1\boldsymbol{\omega}_{\ell}\right|^2}{\left\|\boldsymbol{h}_{2,{\ell}}^H\boldsymbol{\Psi}_{2,{\ell}}\right\|^2 \sigma_{\mathrm{I}}^2+\left\|\boldsymbol{h}_{2,{\ell}}^H\boldsymbol{\Psi}_{2,{\ell}}\boldsymbol{G}\boldsymbol{\Psi}_{1,{\ell}}\right\|^2 \sigma_{\mathrm{I}}^2+\sigma_0^2},
\end{equation}
and the maximum achievable rate in bits/second/Hertz (bps/Hz) of the $\ell$-th user is $R_{\ell} = \frac{1}{L}\log_2\left(1+\gamma_{\ell}\right)$.}
\vspace{-10pt}
\subsection{Problem Formulation}
\vspace{-2pt}
{\color{black}We aim to maximize the minimum achievable rate among all users by optimizing the BS beamforming $\boldsymbol{\Omega}\triangleq\{\boldsymbol{\omega}_{\ell}\}_{\ell=1}^{L}$, active IRSs' amplification factors $\boldsymbol{A}\triangleq\{\boldsymbol{A}_{1,{\ell}},\boldsymbol{A}_{2,{\ell}}\}_{\ell=1}^{L}$, phase shifts $\boldsymbol{\Phi}\triangleq\{\boldsymbol{\Phi}_{1,{\ell}},\boldsymbol{\Phi}_{2,{\ell}}\}_{\ell=1}^{L}$, elements allocation $\boldsymbol{M}\triangleq\{M_1,M_2\}$, and placement $\boldsymbol{X}\triangleq\{x_0,x_1,x_2\}$, formulated as follows.
\allowdisplaybreaks[4]
\begin{align}\vspace{-5pt}
\nonumber \hspace{-0.2cm}\textrm{(P1)}&\max_{\boldsymbol{\Omega},\boldsymbol{A},\boldsymbol{\Phi},\atop \boldsymbol{M},\boldsymbol{X},\eta} \eta \\
&\hspace{-0.4cm}\quad~~\textrm{s.t.}~\quad \eta\leq R_{\ell},\ell=1,\cdots,L,\\
&\hspace{-0.4cm}\qquad\qquad~ P_{\rm B}[\boldsymbol{\Psi}_{1,{\ell}}\boldsymbol{H}_1\boldsymbol{\omega}_{\ell}]_{m_1}^2+\sigma_{\rm I}^2[\boldsymbol{\Psi}_{1,{\ell}}]_{m_1,m_1}^2 \leq P_{\rm e},\nonumber\\
&\quad\qquad\qquad \qquad \quad~~ \ell=1,\cdots,L, m_1=1,\cdots,M_1,\label{cons_powerI1}\\
&\hspace{-0.45cm}\qquad\qquad~ P_{\rm B}[\boldsymbol{\Psi}_{1,{\ell}}\boldsymbol{G}\boldsymbol{\Psi}_{1,{\ell}}\boldsymbol{H}_1\boldsymbol{\omega}_{\ell}]_{m_2}^2+\sigma_{\rm I}^2\left\|[\boldsymbol{\Psi}_{2,{\ell}}\boldsymbol{G}\boldsymbol{\Psi}_{1,{\ell}}]_{m_2,1:M_1}\right\|^2\nonumber\\
&\qquad~~\!\!+\!\sigma_{\rm I}^2[\boldsymbol{\Psi}_{1,{\ell}}]_{m_2,m_2}^2 \!\leq\! P_{\rm e},\!\ell\!\!=\!\!1,\!\cdots\!,L,m_2\!\!=\!\!1,\!\cdots\!,M_2,\label{cons_powerI2}\\
&\hspace{-0.4cm}\qquad\qquad \|\boldsymbol{\omega}_{\ell}\|^2 \leq 1,{\ell}=1,\!\cdots\!,L,\\
&\hspace{-0.4cm}\qquad\qquad \left|[\boldsymbol{\Phi}_{1,{\ell}}]_{m_1,m_1}\right|\!=\!1,{\ell}\!=\!1,\!\cdots\!,L, m_1\!=\!1,\!\cdots\!,M_1,\\
&\hspace{-0.4cm}\qquad\qquad \left|[\boldsymbol{\Phi}_{2,{\ell}}]_{m_2,m_2}\right|\!=\!1, \ell\!=\!1,\cdots,L,m_2\!=\!1,\!\cdots\!,M_2,\label{cons_phaseI2}\\
&\hspace{-0.4cm}\qquad\qquad~ M_1+M_2 = M, M_1 \in \mathbb{N}, M_2 \in \mathbb{N},\label{cons_NumEle}\\
&\hspace{-0.4cm}\qquad\qquad~ x_0+x_1+x_2 = D,\label{cons_x}\vspace{-1cm}
\end{align}
where $[\cdot]_{m,n}$ denotes the $(m,n)$-th entry of a matrix, and the constraints \eqref{cons_powerI1} and \eqref{cons_powerI2} indicate that the amplification factor of each active element is constrained by its per-element amplification power budget, $P_{\rm e}$.}

\vspace{-10pt}
\section{Proposed Solution to Problem (P1)}
\vspace{-2pt}
Problem (P1) is a non-convex optimization problem due to the non-concave objective function, the unit-modulus phase-shift constraints, and the integer elements allocation constraint, thus making it difficult to be solved optimally. 
{\color{black}To address this issue, we propose a two-layer AO algorithm that iteratively optimizes the joint BS and active-IRS beamforming, as well as the active-IRS elements allocation and placement.}

\vspace{-0.3cm}
\subsection{Joint BS and Active-IRS Beamforming Optimization}\vspace{-0.1cm}
{\color{black}First, given the feasible active-IRS elements allocation and placement, the optimization problem (P1) reduces to
\begin{align}
\nonumber \textrm{(P2)}~~~&\max_{\boldsymbol{A},\boldsymbol{\Phi},\hat\eta} \qquad \hat\eta \\
&~~\!\textrm{s.t.}~~~\quad \hat\eta\leq \gamma_{\ell},\ell=1,\cdots,L,\\
&~~\qquad\quad\eqref{cons_powerI1}-\eqref{cons_phaseI2}.\nonumber\vspace{-1cm}
\end{align}}
To solve problem (P2), we first decompose the inter-IRS channel as follows:\vspace{-5pt}
\begin{equation}
    \boldsymbol{G}\!=\!\underbrace{\sqrt{g} \boldsymbol{a}_{\mathrm{r}}\!\left(\theta_{\rm I_1, \rm I_2}^{\mathrm{a}}, \vartheta_{\rm I_1, \rm I_2}^{\mathrm{a}}, M_2\right)}_{\boldsymbol{g}_1} \underbrace{\sqrt{g} \boldsymbol{a}_{\mathrm{t}}^H\!\left(\theta_{\rm I_1, \rm I_2}^{\mathrm{d}}, \vartheta_{\rm I_1, \rm I_2}^{\mathrm{d}}, M_1\right)}_{\boldsymbol{g}_2^H}.\vspace{-5pt}
\end{equation}
Then, we have the following result.
\begin{lemma}\label{lem_oa}
\emph{The optimal solution to problem (P2) is}\vspace{-5pt}
{\color{black}\begin{align}    &\hspace{-0.1cm}\boldsymbol{\omega}_{\ell}\!=\!\boldsymbol{a}_{\mathrm{t}}\left(\theta_{\mathrm{BS}, \mathrm{I}_1}^{\mathrm{d}}, \vartheta_{\mathrm{BS}, \mathrm{I}_1}^{\mathrm{d}}, N\right)\!/\!\left\|\boldsymbol{a}_{\mathrm{t}}\left(\theta_{\mathrm{BS}, \mathrm{I}_1}^{\mathrm{d}}, \vartheta_{\mathrm{BS}, \mathrm{I}_1}^{\mathrm{d}}, N\right)\right\|,\forall \ell,\label{phase_w}\\
    &\hspace{-0.2cm}\phi_{\mathrm{I_1},\ell,m_1}\!=\!\arg \left(\left[\boldsymbol{g}_2\right]_{m_1}\right)\!-\!\arg \left(\left[\boldsymbol{h}_1\right]_{m_1}\right), \forall \ell,m_1,\label{phase_I1}\\
    &\hspace{-0.2cm}\phi_{\mathrm{I_2},l,m_2}\!=\!\arg \left(\left[\boldsymbol{h}_{2,{\ell}}\right]_{m_2}\right)\!-\!\arg \left(\left[\boldsymbol{g}_1\right]_{m_2}\right),\forall \ell,m_2,\label{phase_I2},\\    &\hspace{-0.2cm}a_{\mathrm{I_1},\ell,m_1}=a_1\triangleq\frac{P_{\mathrm{e}}d_0^2}{P_{\mathrm{B}} \beta+\sigma_{\mathrm{I}}^2d_0^2},\forall \ell,m_1,\label{optAFI1}\\
    &\hspace{-0.2cm}a_{\mathrm{I_2},l,m_2}\!=\!\frac{P_{\mathrm{e}}d_0^2 d_1^2}{P_{\mathrm{B}} \!\beta^2\! M_1^2\! a_1^2\!+\!\sigma_{\mathrm{I}}^2 \!\beta\! M_1\! a_1^2\!d_0^2\!+\!\sigma_{\mathrm{I}}^2\!d_0^2\! d_1^2},\forall \ell,m_2.\label{optAFI2}\vspace{-10pt}
\end{align}}
\end{lemma}
\begin{proof}
{The optimal phase shifts are obtained by phase-aligning the double-reflection channel and the optimal amplification factors are obtained by taking the equalities of the power constraints \eqref{cons_powerI1} and \eqref{cons_powerI2}, which can be shown to hold in the optimal solution to problem (P1).}
\end{proof}\vspace{-10pt}
\vspace{-0.1cm}
\subsection{\!Active-IRS\! Elements\! Allocation \!and\! Placement\! Optimization}
\vspace{-0.1cm}
{\color{black}Next, given the optimal active-IRS reflections and BS beamforming in \eqref{phase_w}--\eqref{optAFI2}}, we jointly optimize the active-IRS elements allocation and placement to maximize the minimum achievable rate.
{\color{black}Specifically, by substituting \eqref{phase_w}--\eqref{optAFI2} into \eqref{snr_ori}, the received SNR at the $\ell$-th user is obtained as $\gamma_{\ell}=NP_{\rm B}\beta^3/\xi_{\ell}$}, where
\vspace{-0.1cm}
{\color{black}\begin{align}\label{xi}
    &\xi_{\ell}\triangleq \frac{\sigma_{\mathrm{I}}^4 \sigma_0^2 d_0^2 d_1^2 d_{2,\ell}^2+P_{\mathrm{B}} \sigma_{\mathrm{I}}^2 \sigma_0^2 \beta d_1^2 d_{2,\ell}^2}{P_e^2 M_1^2 M_2^2}+\frac{\sigma_0^2 \sigma_{\mathrm{I}}^2 \beta d_0^2 d_{2,\ell}^2}{P_{\mathrm{e}} M_1 M_2^2}\nonumber \\
    &\quad+\frac{\sigma_{\mathrm{I}}^4 \beta d_0^2 d_1^2+P_{\mathrm{B}} \sigma_{\mathrm{I}}^2 \beta^2 d_1^2}{P_{\mathrm{e}} M_1^2 M_2}+\frac{P_{\mathrm{B}} \sigma_0^2 \beta^2 d_{2,\ell}^2}{P_{\mathrm{e}} M_2^2}+\frac{\sigma_{\mathrm{I}}^2 \beta^2 d_0^2}{M_1 M_2}.
\end{align}}
{\color{black}Based on the above, the solution to problem (P1) can be obtained by solving the following problem (P3).
\begin{align}
\nonumber \textrm{(P3)}~~~&\min_{\boldsymbol{M},\boldsymbol{X},\tilde\eta} \qquad \tilde\eta \\
&~~\!\textrm{s.t.}~~~\quad \tilde\eta\geq\xi_{\ell},\ell=1,\cdots,L,\\
&~~~~~\qquad\eqref{cons_NumEle},\eqref{cons_x}.\nonumber\vspace{-10pt}
\end{align}
Problem (P3) is a non-convex optimization problem, which can be solved by using the AO method in the next.}

\subsubsection{Active-IRS Elements Allocation}
We first optimize the active-IRS elements allocation given fixed active-IRS placement. To tackle the integer constraint \eqref{cons_NumEle}, we relax the discrete values, $\boldsymbol{M}$, into their continuous counterparts, $\boldsymbol{\tilde{M}}\triangleq\{\tilde{M}_1,\tilde{M}_2\}$.
{\color{black}As such, $\xi_{\ell}$ in (15) can be relaxed as 
\begin{align}\label{xi_M}\vspace{-10pt}
    \hspace{-0.1cm}f_{1,{\ell}}(\boldsymbol{\tilde{M}})\triangleq\frac{B_{1,{\ell}}}{\tilde M_1^2 \tilde M_2^2}\!+\!\frac{B_{2,{\ell}}}{\tilde M_1 \tilde M_2^2}\!+\!\frac{B_3}{\tilde M_1^2 \tilde M_2}\!+\!\frac{B_{4,l}}{ \tilde M_2^2}\!+\!\frac{B_5}{\tilde M_1 \tilde M_2},\vspace{-10pt}
\end{align}
where 
\begin{align}
    &B_{1,{\ell}} = \frac{\sigma_\mathrm I^4 \sigma_0^2 d_0^2 d_1^2 d_{2,\ell}^2\!+\!P_{\mathrm{B}} \sigma_\mathrm I^2 \sigma_0^2 \beta d_1^2 d_{2,\ell}^2}{P_e^2}, B_{2,{\ell}} = \frac{\sigma_0^2 \sigma_{\mathrm{I}}^2 \beta d_0^2 d_{2,\ell}^2}{P_{\mathrm{e}}},\nonumber\\\vspace{-10pt}
    &B_3 \!=\! \frac{\sigma_\mathrm I^4 \beta d_0^2 d_1^2+P_{\mathrm{B}} \sigma_{\mathrm{I}}^2 \beta^2 d_1^2}{P_{\mathrm{e}}},B_{4,l} \!=\! \frac{P_{\mathrm{B}} \sigma_0^2 \beta^2 d_{2,\ell}^2}{P_{\mathrm{e}}},B_5 \!=\! {\sigma_\mathrm I^2 \beta^2 d_0^2}.\nonumber
\end{align}}
{\color{black}Next, by introducing the slack variables $\tilde{m_k} = \log(M_k), k=1,2$, $f_{1,{\ell}}(\boldsymbol{\tilde{M}})$ can be re-expressed as
\begin{align}\label{xi_m}
    &\tilde f_{1,{\ell}}(\tilde{m_1},\tilde{m_2})\triangleq B_{1,{\ell}}e^{-2\tilde m_1-2\tilde m_2}+B_{2,{\ell}}e^{-\tilde m_1-2\tilde m_2}\nonumber\\
    &\qquad~~~~~+B_3e^{-2\tilde m_1-\tilde m_2}+B_{4,l}e^{-2\tilde m_2}+B_5e^{-\tilde m_1-\tilde m_2}.\vspace{-10pt}
\end{align}}
{\color{black}Then, the optimal solution to problem (P3) can be obtained by solving the following problem.
\begin{align}
\nonumber \textrm{{(P3.1)}}~~~&\min_{\tilde m_1,\tilde m_2} ~~~~\tilde\eta \\
&~~\textrm{{s.t.}}~~~~\tilde\eta\geq\tilde f_{1,{\ell}}(\tilde{m_1},\tilde{m_2}),\ell=1,\cdots,L,\\
&~\quad\qquad e^{\tilde m_1}+e^{\tilde m_2} \leq M. \label{cons_NumEle3}\vspace{-10pt}
\end{align}}
It can be proved by contradiction that in the optimal solution to problem (P3.1), the equality in constraint \eqref{cons_NumEle3} always holds. As such, problem (P3.1) is a convex optimization problem, which can be efficiently solved by using the interior-point method.
The integer number of reflecting elements can be reconstructed by rounding the continuous solutions to problem (P3.1).

\subsubsection{Active-IRS Placement Optimization}
Next, we optimize the active-IRS placement given fixed active-IRS elements allocation. To this end, we first rewrite $\xi_{\ell}$ in \eqref{xi} as a function of active-IRS locations, given by
{\color{black}\begin{align}
    &\xi_{\ell}=f_{2,{\ell}}(\boldsymbol{D}) \triangleq C_{1} d_0^2 d_1^2 d_{2,\ell}^2+C_{2} d_0^2 d_1^2+C_{3} d_0^2 d_{2,\ell}^2\nonumber\\
    &\qquad\quad\qquad\qquad+C_{4} d_1^2 d_{2,\ell}^2+C_{5} d_0^2+C_{6} d_1^2+C_{7} d_{2,\ell}^2,\label{obj_Lloc}
\end{align}
where $\boldsymbol{D}\triangleq\{d_0,d_1,\{d_{2,\ell}\}_{\ell=1}^{L}\}$,} and 
\begin{align}
    &C_1 \!=\! \frac{\sigma_{\mathrm{I}}^4 \sigma_0^2}{P_e^2 M_1^2 M_2^2},C_2\!=\!\frac{\sigma_\mathrm I^4 \beta}{P_{\mathrm{e}} M_1^2 M_2},C_3\!=\!\frac{\sigma_0^2 \sigma_{\mathrm{I}}^2 \beta}{P_{\mathrm{e}} M_1 M_2^2},\nonumber\\
    &C_4\!=\!\frac{P_{\mathrm{B}} \sigma_{\mathrm{I}}^2 \sigma_0^2 \beta}{P_e^2 M_1^2 M_2^2},C_5\!=\!\frac{\sigma_{\mathrm{I}}^2 \beta^2}{M_1 M_2},C_6\!=\!\frac{P_{\mathrm{B}} \sigma_{\mathrm{I}}^2 \beta^2}{P_{\mathrm{e}} M_1^2 M_2},C_7\!=\!\frac{P_{\mathrm{B}} \sigma_0^2 \beta^2}{P_{\mathrm{e}} M_2^2}.\nonumber
\end{align}
{\color{black}As such, given the feasible active-IRS elements allocation, problem (P3) reduces to \vspace{-5pt}
\begin{align}
\nonumber \textrm{(P3.2)}~~~&\min_{\boldsymbol{D},\boldsymbol{X},\tilde\eta} \qquad \tilde\eta \\
&~~~\!\textrm{s.t.}~\quad \eqref{cons_x},\nonumber\\
&\qquad\quad \tilde\eta\geq f_{2,{\ell}}(\boldsymbol{D}),\ell=1,\cdots,L,\\
&\qquad\quad d_0^2\geq x_0^2+h^2, d_1\geq x_1,\label{cons_Loc1}\\
&\qquad\quad d_{2,\ell}^2\geq (x_2+r_{\ell}\cos{\theta_{\ell}})^2+h^2+(r_{\ell}\sin{\theta_{\ell}})^2,\nn\\
&\qquad\quad \qquad\qquad\qquad\qquad\qquad\quad \ell=1,\cdots,L,\label{cons_Loc2}\vspace{-10pt}
\end{align}
where the equalities in constrains \eqref{cons_Loc1} and \eqref{cons_Loc2} hold in the optimal solution to problem (P3.2).
Note that problem (P3.2) is non-convex due to its non-convex objective function and constraints. To tackle this difficulty, we introduce the slack variables ${y_k} = \log(d_k),k=0,1$, and ${y_{2,{\ell}}} = \log(d_{2,\ell})$, and thus rewrite $f_{2,{\ell}}(\boldsymbol{D})$ as\vspace{-5pt}
\allowdisplaybreaks[4]
\begin{align}\label{xi_y}
    &\hspace{-0.22cm}\tilde f_{2,{\ell}}(\boldsymbol{Y})\triangleq C_1 e^{2 y_0+2 y_1+2 y_{2,{\ell}}}+C_2 e^{2 y_0+2 y_1}+C_3 e^{2 y_0+2 y_{2,{\ell}}}\nonumber\\
    &\qquad\quad~ +C_4 e^{2 y_1+2 y_{2,{\ell}}}+C_5 e^{2 y_0}+C_6 e^{2 y_1}+C_7 e^{2 y_{2,{\ell}}},\vspace{-8pt}
\end{align}
where $\boldsymbol{Y}\triangleq\{y_0,y_1,\{y_{2,{\ell}}\}_{\ell=1}^{L}\}$.}
{\color{black}Problem (P3.2) can then be approximately reformulated as the following convex optimization problem.
\begin{align}\vspace{-8pt}
\nonumber \textrm{{(P3.3)}}~&\min_{\boldsymbol{X},\boldsymbol{Y},\tilde\eta} \quad \tilde\eta\\
&~~\!\textrm{{s.t.}}~~\quad \eqref{cons_x},\nonumber\\
&\qquad\quad~~ \tilde\eta\geq\tilde f_{2,{\ell}}(\boldsymbol{Y}),\ell=1,\cdots,L,\\
&\qquad\quad~~ e^{2\hat y_0}+2e^{2\hat y_0}(y_0-\hat y_0)\geq x_0^2+h^2, \label{cons_y0}\\
&\qquad\quad~~ e^{\hat y_1}+e^{\hat y_1}(y_1-\hat y_1)\geq x_1,\label{cons_y1}\\
&\qquad\quad~~ e^{2\hat y_{2,{\ell}}}\!+\!2e^{2\hat y_{2,{\ell}}}(y_{2,{\ell}}-\hat y_{2,{\ell}})\geq(x_2+r_{\ell}\cos{\theta_{\ell}})^2\nn\\
&\quad\quad~~ \qquad\quad+h^2+(r_{\ell}\sin{\theta_{\ell}})^2, \ell=1,\cdots,L,\label{cons_y2}\vspace{-8pt}
\end{align}
where $\boldsymbol{\hat Y}\triangleq\{\hat y_0,\hat y_1,\{\hat y_{2,{\ell}}\}_{\ell=1}^{L}\}$.}
Note that $\boldsymbol{\hat Y}$ can be initially set as any feasible solution subject to the constraints \eqref{cons_y0}--\eqref{cons_y2}. {\color{black}As such, the objective value of problem (P3.2) decreases with the iterations by using the successive convex approximation (SCA) method.}

{\color{black}To summarize, the inner-layer AO optimizes the active-IRS elements allocation and locations by iteratively solving problems (P3.1) and (P3.2), and the outer layer solves problems (P2) and (P3) iteratively. As such, the objective value of problem (P1) increases in the iterations, which will be shown in Section V to achieve close performance by the exhaustive search.}
\vspace{-10pt}
\section{Capacity Scaling Order Analysis}
\vspace{-2pt}
In this section, we characterize the capacity scaling orders of the double-active-IRS aided wireless system w.r.t. $M$ and $P_{\mathrm{e}}$.

{\color{black}Based on the optimal active-IRS reflections and BS beamforming in \eqref{phase_w}--\eqref{optAFI2}, we first obtain the system capacity scaling orders w.r.t. asymptotically large $M$ for any user in the served area.
Note that the capacity scaling orders can be obtained based on any feasible active-IRS elements allocation and placement with $d_2$ denoting the IRS 2-user distance, since they do not affect the asymptotic result.}
\vspace{-0.2cm}
\begin{proposition}\label{pro1}
\emph{Given fixed $P_{\rm e}$ and any feasible active-IRS elements allocation and placement, the capacity of the double-active-IRS aided wireless system increases with $M$ as $M\to\infty$ according to}\vspace{-0.2cm}
\begin{equation}
    \lim _{M \rightarrow \infty} \frac{C}{\log _2 M}=2.
\end{equation}
\end{proposition}\vspace{-0.15cm}
\begin{proof}
Given the optimal active-IRS reflections and BS beamforming in \eqref{phase_w}--\eqref{optAFI2}, and set $M_1 = \rho M$ and $M_2=(1-\rho)M$ with $0<\rho<1$. When $M\to\infty$, we have
\begin{align}
    &\lim _{M \rightarrow \infty}\xi =\lim _{M \rightarrow \infty} \frac{\sigma_\mathrm I^4 \sigma_0^2 d_0^2 d_1^2 d_2^2+P_{\mathrm{B}} \sigma_\mathrm I^2 \sigma_0^2 \beta d_1^2 d_2^2}{P_e^2 M_1^2 M_2^2}+\frac{\sigma_0^2 \sigma_{\mathrm{I}}^2 \beta d_0^2 d_2^2}{P_{\mathrm{e}} M_1 M_2^2}\nonumber\\
    &\quad+\frac{\sigma_\mathrm I^4 \beta d_0^2 d_1^2+P_{\mathrm{B}} \sigma_{\mathrm{I}}^2 \beta^2 d_1^2}{P_{\mathrm{e}} M_1^2 M_2}+\frac{P_{\mathrm{B}} \sigma_0^2 \beta^2 d_2^2}{P_{\mathrm{e}} M_2^2}+\frac{\sigma_\mathrm I^2 \beta^2 d_0^2}{M_1 M_2}\nonumber\\
    &\quad=\lim _{M \rightarrow \infty}\frac{P_{\mathrm{B}} \sigma_0^2 \beta^2 d_2^2}{P_{\mathrm{e}} \!\left(1\!-\!\rho\right)^2\!\!M^2}+\frac{\sigma_\mathrm I^2 \beta^2 d_0^2}{\rho\left(1\!-\!\rho\right)\!M^2}+o\left(\frac{1}{M^2}\right).\label{pro1_xi}
\end{align}
As such, we have $\lim _{M \rightarrow \infty} \frac{C}{\log _2 M}=2$, which thus completes the proof.
\end{proof}

Proposition \ref{pro1} can be explained as follows.
On one hand, the signal power at the receiver increases with $M$ in the order of $\mathcal{O}(M^4)$ thanks to the multiplicative beamforming gain of double-IRS reflections; while on the other hand, the amplification noise power induced by IRS 1 increases with $M$ in the order of $\mathcal{O}(M^2)$, and that induced by IRS 2 has a scaling order of $\mathcal{O}(M)$, thus leading to the received SNR scaling order of $\mathcal{O}(M^2)$.
Note that the power budget linearly increases with $M$ so that the power amplification gain can be maintained when $M$ is increased.

\begin{proposition}\label{pro2}
\emph{Given fixed $M$ and any feasible active-IRS elements allocation and placement, the capacity of the double-active-IRS aided wireless system does not increase with $P_{\rm e}$ as $P_{\rm e}\to\infty$, i.e.,}
\begin{equation}
    \lim _{P_{\rm e} \rightarrow \infty} \frac{C}{\log _2 P_{\rm e}}=0. 
\end{equation}
\end{proposition}
\begin{proof}
Given the optimal IRS reflections and BS beamforming in \eqref{phase_w}--\eqref{optAFI2}, we have
\begin{align}
    &\lim _{P_{\rm e} \rightarrow \infty}\xi =\lim _{P_{\rm e} \rightarrow \infty} \frac{\sigma_\mathrm I^4 \sigma_0^2 d_0^2 d_1^2 d_2^2+P_{\mathrm{B}} \sigma_\mathrm I^2 \sigma_0^2 \beta d_1^2 d_2^2}{P_e^2 M_1^2 M_2^2}+\frac{\sigma_0^2 \sigma_{\mathrm{I}}^2 \beta d_0^2 d_2^2}{P_{\mathrm{e}} M_1 M_2^2}\nonumber\\
    &\quad+\frac{\sigma_\mathrm I^4 \beta d_0^2 d_1^2+P_{\mathrm{B}} \sigma_{\mathrm{I}}^2 \beta^2 d_1^2}{P_{\mathrm{e}} M_1^2 M_2}+\frac{P_{\mathrm{B}} \sigma_0^2 \beta^2 d_2^2}{P_{\mathrm{e}} M_2^2}+\frac{\sigma_\mathrm I^2 \beta^2 d_0^2}{M_1 M_2}\nonumber\\
    &\quad=\lim _{P_{\rm e} \rightarrow \infty}\frac{\sigma_\mathrm I^2 \beta^2 d_0^2}{M_1 M_2}+o\left(1\right).\label{pro1_pe}\vspace{-8pt}
\end{align}
Based on \eqref{pro1_pe}, it follows that $\lim _{P_{\rm e} \rightarrow \infty} \frac{C}{\log _2 P_{\rm e}}=0$, which thus completes the proof.\vspace{-0.1cm}
\end{proof}

Proposition \ref{pro2} can be explained as follows.
First, the beamforming gain does not increase since it only depends on the number of reflecting elements, which remains unchanged when the active-IRS per-element amplification power increases.
Second, both the amplification noise power and desired signal power linearly increase with the per-element amplification power, hence resulting in no increase in the received SNR at the user.

\begin{remark}
\emph{\textbf{(Double-active-IRS versus single-active-IRS)}
Given fixed per-element amplification power, the received SNR scaling order with double active IRSs is higher than that with a single active IRS in terms of $M$, i.e., $\mathcal{O}(M^2)$ versus $\mathcal{O}(M)$, with the latter shown in \cite{9377648}. This is due to the higher beamforming gain provided by double reflections in the double-IRS-aided system compared to the single-IRS case.}
\end{remark}

\begin{remark}
\emph{\textbf{(Double-active-IRS versus double-passive-IRS)} As shown in \cite{9060923}, the received SNR of the double-passive-IRS aided system increases with $M$ in the order of $\mathcal{O}(M^4)$, which is higher than the double-active-IRS case, i.e., $\mathcal{O}(M^2)$. Note that although the received signal power increases with $M$ in the order of $\mathcal{O}(M^4)$ for both active and passive IRSs, the amplification noise power in the active-IRS case also increases in the order of $\mathcal{O}(M^2)$, hence reducing its SNR scaling order. However, it is worth mentioning that the double-active-IRS aided system is expected to achieve superior rate performance over its passive counterpart in practice, especially when the number of reflecting elements is moderate and/or the amplification power is large, due to the additional signal power amplification gain.}
\end{remark}

\vspace{-15pt}
\section{Numerical Results}
\vspace{-2pt}
Numerical results are presented in this section. The horizontal distance from the BS to the user is set as $D=200$ m and the IRSs are deployed at an altitude of $H=5$ m. {\color{black}We assume that $L=4$ users are randomly distributed within the radius of 30 m from the user-zone center.}
If not specified otherwise, we set the per-element amplification power as $P_{\rm e}=1$ mW, and the fixed elements budget as $M=128$.
{\color{black}Other parameters are set as $N=4$, $\lambda=0.4$ m, $\beta=(\lambda / 4 \pi)^2=-30 \mathrm{~dB}$, $\alpha=2$, $P_{\rm B}=1$ W, and $\sigma_0=\sigma_{\rm I}=-80$ dBm.}

\begin{figure}[t] \centering
{\vspace{-5pt}\subfigure[] {
\includegraphics[width=0.46\linewidth]{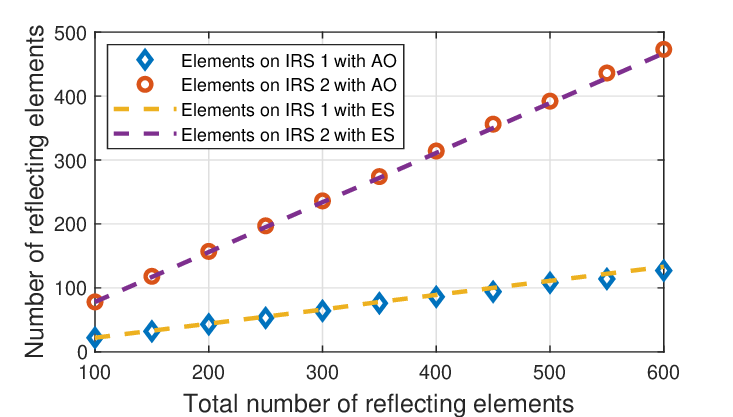}  
}}
{\vspace{-5pt}\subfigure[] {
\includegraphics[width=0.46\linewidth]{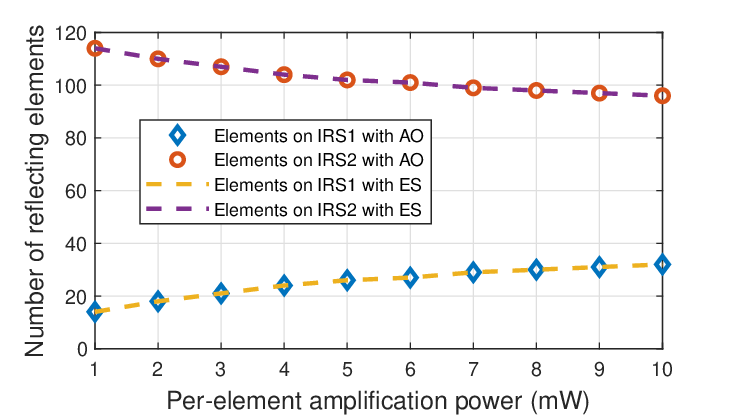}
}}
{\caption{Optimal active-IRS elements allocation.}\vspace{-13pt}\label{FigM}}
\end{figure}

\begin{figure}[t] \centering
{\vspace{-5pt}\subfigure[] {
\includegraphics[width=0.46\linewidth]{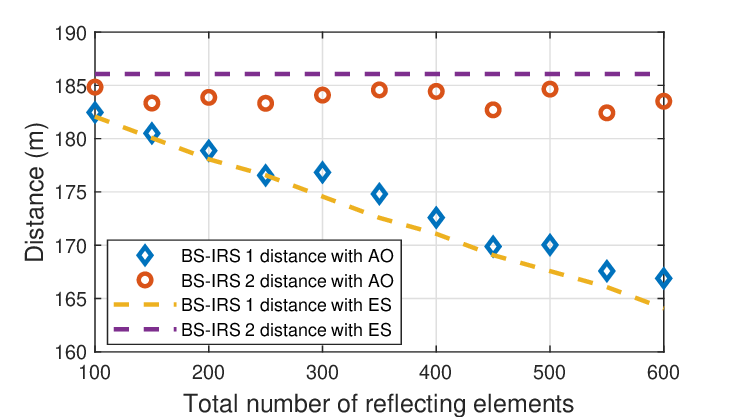}  
}}
{\vspace{-5pt}\subfigure[] {
\includegraphics[width=0.46\linewidth]{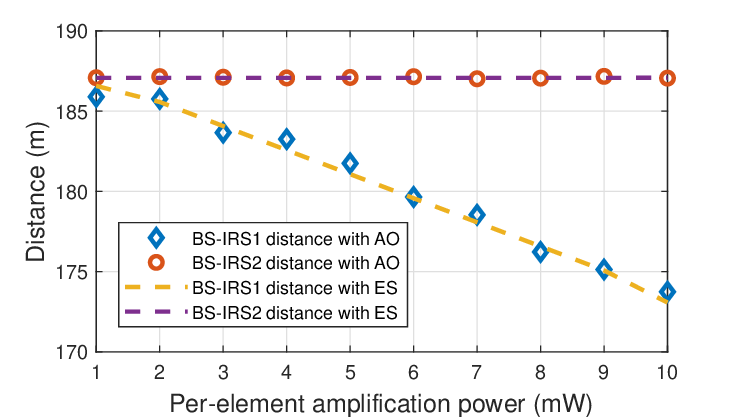}
}}
{\caption{Optimal active-IRS locations.}\vspace{-18pt}\label{FigLOC}}
\end{figure}

We first compare the optimized active-IRS elements allocation and placement by using the proposed AO algorithm and the exhaustive search (ES).
Fig. \ref{FigM} shows the optimized elements allocation for the two IRSs.
One can observe that the proposed AO algorithm yields close-to-optimal elements allocation with the ES method.
Besides, in Fig. \ref{FigM}(a), we observe that the optimized number of elements for both IRSs increases with the total number of reflecting elements, while more elements are allocated to IRS 2 (closer to user) to reduce the amplification noise induced by IRS 1.
Moreover, as observed in Fig. \ref{FigM}(b), with an increasing per-element amplification power, the number of reflecting elements allocated to IRS 1 increases for optimally balancing the trade-off between the signal and noise power amplification.
In Fig. \ref{FigLOC}, we compare the optimized active-IRS locations.
We observe that as the total elements budget or per-element amplification power increases, IRS 1 should be placed even closer to the BS for minimizing the amplification noise power with smaller amplification factors, while IRS 2 should be placed above the user to minimize the cascaded path loss.

\begin{figure}[t] \centering
{\subfigure[Rate versus total number of reflecting elements.] {
\includegraphics[width=0.46\linewidth]{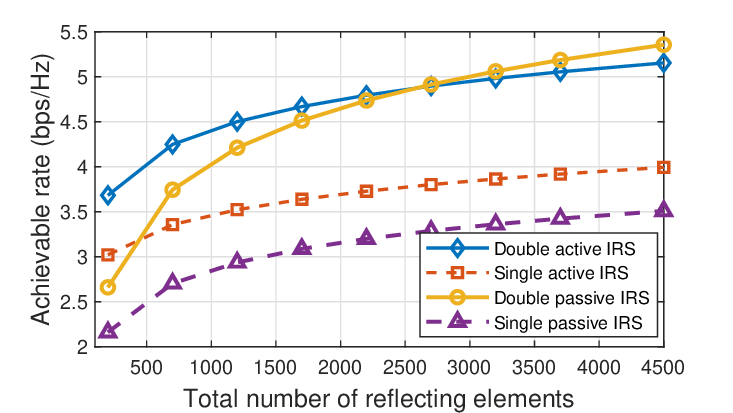}  
}}
{\subfigure[Rate versus per-element amplification power.] {
\includegraphics[width=0.46\linewidth]{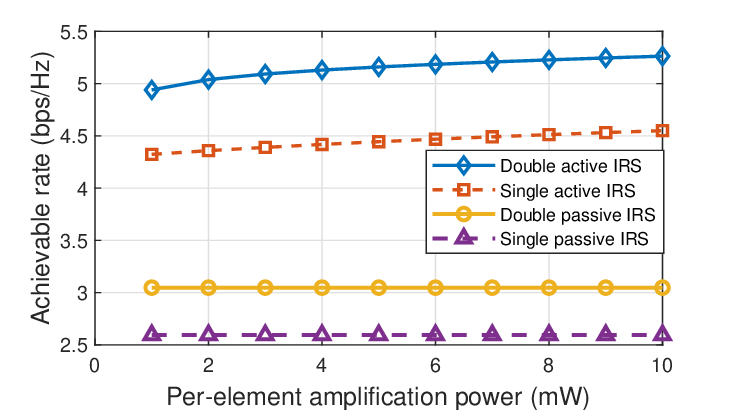}
}}
{\caption{Rate comparison of double-active-IRS with benchmark systems.}\vspace{-13pt}\label{FigC}}
\end{figure}

Next, Fig. \ref{FigC} compares the achievable rates of wireless communication system aided by double active IRSs, a single active IRS, double passive IRSs, and a single passive IRS.
Fig. \ref{FigC}(a) plots the rate performance versus the total number of reflecting elements. 
First, it is observed that the double-active-IRS aided system achieves a higher capacity scaling order than that aided by a single active IRS owing to the higher double-reflection multiplicative beamforming gain.
Second, the double-active-IRS aided system achieves a higher rate than its double-passive-IRS counterpart when the number of reflecting elements is not excessively large thanks to the additional power amplification gain.
In Fig. \ref{FigC}(b), the rate performance is compared versus the per-element amplification power. 
It is observed that the active-IRS cases outperform the passive-IRS cases thanks to the power amplification gain. 
Moreover, the achievable rates of active-IRS aided systems increase slowly with the per-element amplification power, despite the non-favorable asymptotic result in Proposition \ref{pro2}.
This is because both the powers of the desired signal and amplification noise linearly increase with the per-element amplification power.
\vspace{-8pt}
\section{Conclusions}\vspace{-4pt}
{\color{black}In this letter, we proposed an AO-based algorithm to optimize the BS beamforming, active-IRS reflections, elements allocation, and placement for maximizing the achievable rate of the double-active-IRS aided wireless communication system.}
Besides, we characterized the capacity scaling orders of the double-active-IRS aided wireless system w.r.t. the number of reflecting elements and the per-element amplification power. 
It was shown that given the fixed per-element amplification power, the received SNR increases asymptotically with the square of the number of reflecting elements; while given the fixed number of reflecting elements, it does not increase with the per-element amplification power when it is asymptotically large.
Numerical results were presented to evaluate the proposed algorithm and compare the rate performance of the double-active-IRS aided wireless system with various benchmark systems.
\vspace{-5pt}
\vspace{-0.2cm}
\bibliographystyle{IEEEtran}
\bibliography{ref.bib}
\end{document}